\documentclass[12pt]{article}
\usepackage{amsmath, amsfonts, amssymb, amsthm, graphicx, geometry, arydshln, umoline, subfig, newtxtext, newtxmath, color, comment, soul, enumerate, bm}
\usepackage{appendix}
\usepackage{tikz}
\usetikzlibrary{positioning}
\usepackage{natbib}
\usepackage[colorlinks=true,citecolor=blue]{hyperref}
\usepackage{booktabs}
\usepackage{makecell}
\usepackage{bbm}

\theoremstyle{definition}
\newtheorem{theorem}{Theorem}

\newtheorem{lemma}{Lemma}

\newtheorem{proposition}{Proposition}

\begin{document}

\title{Peak-Robust Voting Rules\thanks{We thank Noriaki Kiguchi, Yukio Koriyama, Herv\'e Moulin, Mat\'ias N\'u\~nez, Yuki Tamura, William Thomson, Marc Vorsatz, and seminar participants at Waseda university, the spring meeting on the Operations Research Society of Japan 2026, I France-Spain Meeting on Microeconomic Theory, and SSCW 2026 for helpful comments. 
Satoshi Nakada acknowledges the financial support by JSPS KAKENHI Grant Number 25K16606 and 25K00618.
All the remaining errors are our own.}
}
\author{
Satoshi Nakada\thanks{School of Management, Department of Business Economics, Tokyo University of Science, 1-11-2, Fujimi, Chiyoda-ku, Tokyo, 102-0071, Japan. Email: snakada@rs.tus.ac.jp}
\and 
Toshiya Yoshimura\thanks{Undergraduate School of Management, Department of Business Economics, Tokyo University of Science, 1-11-2, Fujimi, Chiyoda-ku, Tokyo, 102-0071, Japan. Email: tossyyossy01@gmail.com}
}
\date{\today}
\maketitle

\begin{abstract}
This paper proposes new robustness criteria for social choice correspondences under single-peaked preferences, inspired by the concepts of robustness in statistical estimation, where robust estimators are designed to be resilient to both model misspecification and outliers. 
Motivated by robustness to model assumptions, we introduce \textit{peak-robustness}: a voting rule is peak-robust if it never selects an alternative that is a majority loser relative to some unchosen alternative for any preference profile sharing the same peak profile. 
To capture robustness to outliers, we propose \textit{tail-invariance}, which requires that variations in the tails of the peak distribution do not affect the collective decision.
Our main result shows that the median voting rule is the unique efficient rule satisfying these robustness criteria. When peak-robustness is weakened, we characterize the broader class of \textit{quantile rules}. 
Taken together, these results provide a robustness-based axiomatic foundation for median and quantile voting rules, independent of the traditional strategy-proofness approach.

\end{abstract}

\section{Introduction}

\subsection{Motivation}
Consider a voting problem in which individuals have single-peaked preferences over a linearly ordered set of alternatives \citep{black1948rationale}. 
A central question in this environment is how to aggregate the information contained in voters' peaks into a collective decision. The classical literature has approached this problem primarily from the perspective of strategic-robustness. 
In particular, \citet{Moulin1980} showed that strategy-proof voting rules on the single-peaked domain coincide with generalized median voter schemes, thereby establishing median-type rules as canonical mechanisms for aggregating preferences.\footnote{There is a vast literature on the design of strategy-proof voting rules, initiated by the Gibbard--Satterthwaite impossibility theorem \citep{gibbard1973manipulation,satterthwaite1975strategy}. The single-peaked domain provides a classical escape from this impossibility, and the characterization of generalized median voter schemes is one of the best-known positive results in this literature. For more recent developments, see \citet{saporiti2009strategy}, \citet{achuthankutty2018single}, \citet{hagiwara2026strategy}, and the references therein.}

This paper takes a different perspective. Rather than asking which voting rules are immune to strategic manipulation, we ask which voting rules are robust to uncertainty in the informational content of preference profiles. 
Specifically, we view the profile of peaks as a datum describing the voters' opinions and interpret a voting rule as a statistical estimator for a representative opinion.
This perspective naturally suggests importing ideas from robust statistics \citep{Huber1981}, where estimators are designed to remain stable in the presence of model uncertainty and extreme observations. 
We ask whether analogous robustness principles can provide a foundation for desirable voting rules.

To formalize this viewpoint, we introduce robustness criteria requiring that the outcomes of voting rules should not depend excessively on unobserved or extreme aspects of preferences. The first criterion is \textit{peak-robustness}. Intuitively, this criterion requires that the chosen outcome can never be regarded as a majority loser, regardless of how voters rank alternatives away from their peaks. While peak-only voting rules are often attractive because they require only limited information from voters, the planner still does not observe the full preference profile. Consequently, many underlying preference profiles may be consistent with the same profile of peaks. Peak-robustness adopts a conservative attitude toward this uncertainty by requiring that the chosen outcome never be a majority loser for \emph{any} preference profile consistent with the observed peaks. We also consider a weaker version, referred to as \textit{weak peak-robustness}, which requires only that the chosen alternative not lose to another alternative for \emph{all} preference profiles consistent with the observed peaks.

Our second robustness criterion is \textit{tail invariance}. This axiom requires that perturbations of the most extreme peaks do not affect the outcome as long as the central structure of the peak profile remains unchanged. In other words, changes in the tails of the peak distribution should have no impact on the collective decision. This requirement directly reflects the central idea of robust statistics that an estimator should not be influenced by extreme observations or outliers.

Our first result establishes that, together with efficiency, peak-robustness and tail invariance characterize the median voter rule. 
Importantly, we neither restrict attention to peak-only voting rules nor impose strategy-proofness, both of which play central roles in the classical literature. Thus, our characterization provides a new foundation for the median voter rule based on statistical robustness considerations rather than strategic robustness.

While peak-robustness uniquely identifies the median voter rule, weak peak-robustness admits a richer class of set-valued voting rules. Our second result shows that, together with efficiency, anonymity, endpoint invariance, and peak monotonicity, weak peak-robustness characterizes the class of \textit{quantile rules}. These rules can be interpreted as set-valued generalizations of the median voter rule and select outcomes determined by order statistics of the peak distribution, much like statistical quantiles.

We emphasize that our perspective differs from the classical strategy-proofness approach to single-peaked preferences. Beginning with the seminal work of \citet{Moulin1980}, a vast literature has characterized generalized median voter schemes through strategic robustness perspective.
In contrast, we do not impose strategy-proofness or any related incentive requirement. Instead, we investigate which voting rules emerge when robustness is taken as the primary design objective. Although both approaches are motivated by robustness concerns in a broad sense, the robustness notions studied here are logically independent of strategy-proofness.

Taken together, our results provide a robustness-based foundation for median and quantile voting rules. While the classical literature derives median-type rules from strategic considerations, we show that the same rules arise from robustness requirements concerning incomplete preference information and extreme opinions. In this sense, the paper establishes a new connection between social choice theory and robust statistics.

The rest of this paper is organized as follows.
In Section \ref{sec_voting}, we introduce the formal model and voting rules.
In Section \ref{sec_robust}, we provide our main robustness axioms.
In Section \ref{sec_main}, we provide our main results.
Section \ref{sec_conclusion} discusses several related issues as concluding remarks.
All omitted proofs are relegated to the Appendix.

\subsection{Related literature}
Our paper is closely related to recent studies of social choice correspondences on the single-peaked domain. \citet{klaus2020strategy} and \citet{bhattacharya2024strategy} study choice correspondences that are assumed a priori to select intervals of alternatives. Each paper characterizes a class of generalized median voting rules under this interval restriction, which differs from our approach. In contrast, we do not impose an interval structure; rather, it emerges endogenously as a consequence of our axioms.

The robustness perspective has also been explored in voting and institutional design. 
\citet{nakada2025robust} propose robustness criteria for voting rules with binary alternatives that avoid worst-case scenarios in which a voting rule is Pareto dominated by another rule under the true but unknown distribution of preferences. They characterize weighted majority rules through these robustness requirements. \citet{kiguchi2025robust} extend this approach to environments with more than two alternatives and single-peaked preferences. 
Using a worst-case utilitarian criterion, they show that the median voter rule maximizes the objective function, thereby providing a robustness-based foundation for the median rule. Our approach shares the same spirit in that desirable voting rules should avoid unfavorable worst-case scenarios. 
A key difference, however, is that our robustness criteria are motivated by a statistical perspective and do not rely on incomplete information, probabilistic uncertainty, or strategic considerations.

More broadly, our paper is related to a growing literature connecting social choice theory and statistics. In a seminal contribution, \citet{pivato2013voting} interprets voting rules as statistical estimators and shows that several classical voting rules can be viewed as maximum-likelihood or maximum a posteriori estimators under suitable probabilistic models. In particular, the median rule emerges from an exponential error model on a metric space. Beyond their classical role as robust statistical estimators, quantiles have also been axiomatized as fundamental decision criteria \citep{rostek2010quantile,chambers2007ordinal,chambers2009axiomatization,fadina2023one} and have found applications in dynamic choice, portfolio selection, and statistical decision theory \citep{de2019dynamic,castro2022portfolio,manski2023statistical}. Our results complement these studies by providing an axiomatic foundation for median and quantile aggregation in collective decision making. Unlike the statistical-estimation approach, our characterization does not rely on probabilistic assumptions or error models; instead, quantile rules emerge directly from robustness requirements imposed on social choice correspondences.

\section{Voting rules}\label{sec_voting}
\subsection{Median voter rules}
The set of voters is $N=\{1,\ldots,n\}$, where $n=2m+1\ (m\ge 1)$ is an odd number, and the set of alternatives is  $X=\{a_1,\ldots,a_K\}$. The alternatives are arranged according to an ordering $<$ on X such that $a_1<a_2<\cdots<a_K$. 
Each voter $i$ has a strict preference, $P_i$ (a complete, transitive, and asymmetric binary relation on $X$).
The weak preference $R_i$ is induced by $P_i$ such that  $x R_i y $ if and only if $ x P_i y \lor x = y$.
We also assume that $P_i$ is \textbf{single-peaked} on $X$, i.e., there exists a ``peak", $\tau(P_i)$, such that for any $x,y\in X$,
\[
[y<x\le \tau(P_i) \text{ or } \tau(P_i)\le x < y]\Rightarrow xP_i y,
\]
where the peak, $\tau(P_i)$, is the top-ranked alternative in $X$ for voter $i\in N$. 
For $A\in 2^X\setminus\{\emptyset\}$, let $b_{R_i}(A)$ be the best alternative in $A$ for $i$, i.e., for any $a\in A$, $b_{R_i}(A)R_i a$. 
Similarly, for $A\in 2^X\setminus\{\emptyset\}$, let $w_{R_i}(A)$ be the worst alternative in $A$ for $i$, i.e., for any $a\in A$, $a R_i w_{R_i}(A)$.

Let $\mathcal{SP}$ be the set of all single-peaked preferences over $X$ according to $<$ and let $P=(P_1,\ldots,P_n)$ denote a profile of single-peaked preferences where each $P_i \in \mathcal{SP}$. 
Let $\mathcal{SP}^n$ be the set of all single-peaked profiles on $X$. 
For any $P \in \mathcal{SP}^n$, let $\tau(P)=\{ \tau(P_1), \ldots, \tau(P_n) \}$ be the profile of peak at $P$.
We denote the $p$-th smallest peak in $\tau(P)$ by $\tau_{(p)}(P)$, where $\tau_{(1)}(P) \leq \tau_{(2)}(P) \leq \dots \leq \tau_{(n)}(P)$ where $p \in \{1, \dots, n\}$.
In particular, let $\underline \tau(P)=\text{min}_{i \in N}\tau(P_i)=\tau_{(1)}(P)$ and $\overline \tau(P)=\text{max}_{i \in N}\tau(P_i)=\tau_{(n)}(P)$ be the lowest and largest peak in  $\tau(P)$, respectively, where the minimum and maximum are taken according the ordering $<$ on $X$.

A social choice function (SCF) $f: \mathcal{SP}^n \to X$ produces an alternative $a \in X$ for every profile $P \in \mathcal{SP}^n$.
A well-known SCF is a \textbf{generalized median voter rule} introduced by  \citet{Moulin1980}, which is defined as follows: There exists $n-1$ alternatives $x_1, \ldots, x_{n-1}$ such that, for any $P \in \mathcal{SP}^n$, 
\[
F^{x}(P)= \text{med}\{ \tau(P_1), \ldots, \tau(P_n), x_1,\ldots, x_{n-1}\}.
\]
In particular, if $x_j \in \{a_1,a_K\}$ for any $j=1,\ldots, n-1$ and $|\{j \in N\mid x_j=a_1\}|= |\{j \in N\mid x_j=a_K\}|$, then it is the \textbf{median voter rule}, that is, $F(P)=\text{med}\{ \tau(P_1), \ldots, \tau(P_n)\}$.
A remarkable property of the median rule is to select the \textbf{Condorcet winner} of $P$, denoted by $C(P)$.
Indeed, since $n$ is odd, it is known that $C(P)=\text{med}\{\tau(P)\}$ for any $P \in \mathcal{SP}^n$.
In the simgle peak-domain, \citet{Moulin1980} provides the following axiomatic characterization for the class of generalized median voter rules relying on the \textit{strategy-proofness}.

\begin{proposition}\label{prop_Moulin}\citep{Moulin1980}
 An SCF,  $f: \mathcal{SP}^n \to X$,  satisfies \textit{anonymity}, \textit{efficiency}, \textit{strategy-proofness} if and only if $f$ is a generalized median voter rule, where the axioms are defined as follows.
 
 \begin{itemize}
\item[] \textbf{Anonymity}: For any $P\in \mathcal{SP}^n$ and permutation $\sigma$ of $N$, $f(P)=f(P^\sigma)$, where $P^\sigma=(P_{\sigma(1)}, \ldots,P_{\sigma(n)})$.

\item[] \textbf{Efficiency}: For any $P\in \mathcal{SP}^n$ and any $x\in X$,
\[
[\exists y\in X,\ \forall i\in N,\  y P_i x] \Rightarrow x \neq f(P) .
\] 

\item[] \textbf{Strategy-proofness}: For any $i \in N$, and any $(P_i,P_{-i})\in \mathcal{SP}^n$, 
\[
f(P_i,P_{-i})R_i f(P_i^\prime,P_{-i})\  \forall P_i^\prime\in\mathcal{SP}.
\]
\end{itemize}
\end{proposition}

\subsection{The generalized median correspondences}
A social choice correspondence (SCC), $F: \mathcal{SP}^n \to 2^X \setminus \{\emptyset\}$ chooses a set $A\in 2^X \setminus \{\emptyset\}$ for every profile $P\in \mathcal{SP}^n$.
Note that an SCF is an SCC with $|F(P)| = 1$ for any $P \in \mathcal{SP}^n.$ 
For each $P\in \mathcal{SP}^n$, let $l(P),u(P)\in F(P)$ be the minimum and maximum outcomes, respectively, that are chosen at $P$, which are defined as 
\[
 \ l(P)\le a \ \text{ and }\  \ u(P)\ge a, \forall a \in F(P).
\]

We introduce the following new rules that generalize the median voter rule.
To begin with, we provide the definition of an ``interval'' of alternatives.
For any alternative $a,b \in X$, we define $[a,b]$ as follows:
\[
[a,b] = \{x \in X \mid a \le x \le b\}.
\]
An SCC, $F^{Q_L,Q_U}: \mathcal{SP}^n \to 2^X \setminus \{\emptyset\}$, is a \textbf{quantile rule} if there exist $Q_L,Q_U\in \{1,\ldots, m\}$ such that for any $P\in \mathcal{S}^n$,
\[
F^{Q_L,Q_U}(P)=[\tau_{(Q_L+1)}(P), \tau_{(n-Q_U)}(P)].
\]
If $Q_L = Q_U$, we call this a \textbf{symmetric quantile rule}.

Although our setting differs slightly, the quantile rule can be interpreted as a special case of \textit{the (anonymous and efficient) generalized median correspondence} introduced in \citet{klaus2020strategy}. 
Formally, an SCC $F^{\alpha,\beta}: \mathcal{SP}^n \to 2^X \setminus \{\emptyset\}$, is a \textbf{(anonymous and efficient) generalized median correspondence} if there exist $2(n-1)$ fixed alternatives $(\alpha,\beta)=(\alpha_1,\ldots,\alpha_{n-1},\beta_1,\ldots,\beta_{n-1}) \in X^{2(n-1)}$ such that for any $P\in \mathcal{S}^n$,
\[
  F^{\alpha,\beta}(P)=[F^{\alpha}(P), F^{\beta}(P)],
\]
where for all $i \in \{1,\ldots,n-1\}$,  $\alpha_i \le \beta_i$.
If for all $i \in \{1,\ldots,n\}$, $\alpha_i, \beta_i \in \{a_1,a_K\}$ and $m \le |\{j \in N\mid \alpha_j=a_1\}| $ and $m \le |\{j \in N\mid \beta_j=a_K\}|$, then an (anonymous and efficient) generalized median correspondence $F^{\alpha,\beta }$ is a quantile rule.
\cite{klaus2020strategy} provides a characterization of generalized median correspondences under the assumption that outcomes are always intervals, which we refer to as \textit{convexity}, as a generalization of Propostion \ref{prop_Moulin}.

\begin{proposition}\citep{klaus2020strategy}
An SCC, $F$, satisfies \ textit {convexity}, \textit{anonymity}, \textit{efficiency}, and \textit{optimistic and pessimistic strategy-proofness} if and only if $F$ is a (anonymous and efficient) generalized median correspondence, where the axioms are defined as follows

\begin{itemize}
\item[] \textbf{Anonymity}: For any $P\in \mathcal{SP}^n$ and permutation $\sigma$ of $N$, $F(P)=F(P^\sigma)$, where $P^\sigma=(P_{\sigma(1)}, \ldots,P_{\sigma(n)})$.
 
\item[] \textbf{Efficiency}: For any $P\in \mathcal{SP}^n$ and any $x\in X$,
\[
[\exists y\in X,\ \forall i\in N,\  y P_i x] \Rightarrow x \notin F(P) .
\] 

\item[] \textbf{Convexity}: A SCC, $F$, is said to be \emph{convex} if for any $P\in \mathcal{SP}^n$, and for any $a,b,c\in X$ where $a<b<c$,
\[
[a\in F(P) \ \land\  c\in F(P)] \Rightarrow b\in F(P).
\] 

\item[] \textbf{Optimistic and pessimistic strategy-proofness}: For any $ i \in N$, and any $(P_i,P_{-i})\in \mathcal{SP}^n$, 
\[
b_{P_i}(F(P_i,P_{-i}))R_i b_{P_i}(F(P_i^\prime,P_{-i})) \text{ or } w_{P_i}(F(P_i,P_{-i}))R_i w_{P_i}(F(P_i^\prime,P_{-i}))\  \ \forall P_i^\prime\in\mathcal{SP}.
\]
\end{itemize}
 
\end{proposition}

\section{Robustness}\label{sec_robust}
As discussed in the previous section, the axiomatic analysis of SCCs in the single-peaked domain has primarily focused on the design of rules that are robust to strategic manipulation. 
In contrast, we investigate several robustness criteria grounded in the concept of statistical robustness and aim to design a voting rule that satisfies these criteria.

The following two robustness criteria are the key concepts of our analysis and were informally illustrated in the Introduction.
\textit{Peak-robustness} requires that the chosen outcome not be a majority loser for \emph{any} preference profile consistent with the observed peaks, thereby ruling out the worst-case scenario.
As a weaker requirement, \textit{weak peak-robustness} requires only that the chosen alternative not lose to another alternative for \emph{all} preference profiles consistent with the observed peaks, thereby ruling out only more severe worst-case scenarios.
Our robustness criteria are conceptually inspired by the logic of the Condorcet committee introduced by \cite{fishburn1981} and \cite{Gehrlein1985}.\footnote{See also \citet{kamwa2018coincidence} for a recent study in this topic.}
Condorcet committee is defined as a committee of a fixed size that is preferred by a majority of voters to every other committee of the same size. 
While our research departs from this model in that it does not seek to elect a committee of a predetermined size, it shares the same normative commitment to the majority principle as a benchmark for collective choice.

\begin{itemize}
\item[] \textbf{Peak-robustness}:  For any $P\in \mathcal{SP}^n$ and $x\in X$, if there exist $a \in F(P)$ and $P^\prime \in \mathcal{SP}^n$ with $\tau(P')=\tau(P)$ such that $ |\{i\in N\mid xP_i^\prime a\}| > |\{i\in N\mid aP_i^\prime x\}|$, then $x\in F(P)$.

\item[] \textbf{Weak peak-robustness}:  For any $P\in \mathcal{SP}^n$ and $x\in X$, if there exist $a \in F(P)$ such that $ |\{i\in N\mid xP_i^\prime a\}| > |\{i\in N\mid aP_i^\prime x\}|$ for any $P^\prime \in \mathcal{SP}^n$ with $\tau(P')=\tau(P)$, then $x\in F(P)$.
\end{itemize}

In the strict and single-peaked preference domain, since an undominated alternative indeed dominates other alternatives in terms of the pairwise-majority criterion, (weak) peak-robustness can be interpreted as the criterion of a robust Condorcet committee.
We show the following implication of the robustness criteria, which is useful to show our main results.

\begin{lemma}\label{lem_condorcet_convex}
$F$ satisfies \textit{weak peak-robustness} if and only if $F$ satisfies \textit{set-valued Condorcet consistent} and \textit{convex}, where the former axiom is defined as follows.

\begin{itemize}
\item[] \textbf{Set-valued condorcet consistency}: A SCC, $F$, is said to be \emph{set-valued condorcet consistent} if for any $P\in \mathcal{SP}^n$, 
\[
C(P)\in F(P).
\]

\end{itemize}
\end{lemma}

\begin{proof}
We first show if part. Take any $P \in \mathcal{SP}^n$.
First, we show that $F$ satisfies \textit{set-valued Condorcet consistency}.
Since $F(P) \neq \emptyset$, for any $a \in F(P)\setminus \{C(P)\}$ and $P' \in \mathcal{SP}^n$ with $\tau(P')=\tau(P)$, we must have $|\{i \in N \mid C(P) P'_i a\}|=|\{i \in N\mid \text{med}\{\tau(P')\} P'_i a\}| > |\{i \in N \mid a P'_i \text{med}\{\tau(P')\}\}|=|\{i \in N\mid a P'_i  C(P)\}| $.
Therefore, by \textit{weak peak-robustness}, we have $C(P) \in F(P)$, so that \textit{set-valued Condorcet consistency} holds.

Next, we show that $F$ satisfies \textit{convexity}.
Suppose that $a,c \in F(P)$ with $a<b<c$.
If $a<c\le C(P)$, then $b P_i a$ for any $i \in N$ with $\tau(P_i) \ge C(P)$.
Since $C(P)=\text{med}\{\tau(P)\}$, we have  $|\{i \in N \mid b P'_i a\}| \ge m+1 > |\{i \in N \mid a P'_i b\}|$ for any $P' \in \mathcal{SP}^n$ with $\tau(P')=\tau(P)$.
Hence, by \textit{weak peak-robustness}, we have $b \in F(P)$.
Symmetrically, if $C(P) \le a<c$, we can see that $b \in F(P)$.
If $a<C(P)<c$, note that $b\le C(P)$ or $b>C(P)$.
Hence, for any $P' \in \mathcal{SP}^n$ with $\tau(P')=\tau(P)$,  by the same arguments in the above, we have  $|\{i \in N \mid b P'_i a\}| > |\{i \in N \mid a P'_i b\}|$ in the former case and $|\{i \in N \mid b P'_i c\}| > |\{i \in N \mid c P'_i b\}|$ in the latter case.
Therefore, by \textit{weak peak-robustness}, we have $b \in F(P)$ for both cases.
Summarizing the arguments,  \textit{convexity} holds.

We now prove only if part.
Suppose that $F$ satisfies \textit{set-valued Condorcet consistency}  and \textit{convexity} but $F$ is not \textit{weak peak-robust}.
If $F$ does not satisfy \textit{weak peak-robustness}, then there exist  $P \in \mathcal{SP}^n$ and  $x \in X$ such that there exists $a \in F(P)$ such that $|\{i\in N\mid xP_i^\prime a\}| > |\{i\in N\mid aP_i^\prime x\}|$ for any $P^\prime \in \mathcal{SP}^n$ with $\tau(P')=\tau(P)$,  and yet $x\notin F(P)$.
If  $a=C(P)=\text{med}\{\tau(P)\}$, then  there is no $x \in X$ such that $|\{i \in N \mid x P_i a \}| > |\{i \in N \mid a P_i x\}|$ by the definition of the \textit{Condorcet winner}.
Therefore, such $a \neq C(P)$.

Without loss of generality, assume that $a<C(P)$.
If $x<a$, by the same arguments in the proof of if part, then $|\{i \in N \mid x P_i a\} | < |\{i \in N \mid a P_i x\}|$ for all $P \in \mathcal{SP}^n$, so there is no $x (<a)$ such that $|\{i \in N \mid x P_i a\} | > |\{i \in N \mid a P_i x\}|$ for all $P \in \mathcal{SP}^n$.
If $x>C(P)$, there exists $P^\prime$ such that for all $i \in N$ whose peaks are to the left of or equal to $C(P)$, $a P_i x$.
In this case, $|\{i \in N \mid x P_i a \}| < m+1 \le |\{i \in N \mid a P_i x\}|$, so  there is no $x (<a)$ such that $|\{i \in N \mid x P_i a \}| > |\{i \in N \mid a P_i x\}|$ for all $P \in \mathcal{SP}^n$.
Therefore, such $x$ is located between $a$ and $C(P)$, i.e., $a \le x \le C(P)$.
However, $F$ satisfies \textit{set-valued Condorcet consistency} and \textit{convexity}, so if $a \in F(P)$, for all alternatives that are located between $a$ and $C(P)$ also belong to $F(P)$.
Therefore, if such $x$ is located between $a$ and $C(P)$, this violates \textit{set-valued Condorcet consistency} and \textit{convexity}.
Hence, $F$ satisfies \textit{weak peak-robustness}.
\end{proof}

The following axiom, \textit{tail-invariance}, is a new concept to capture the idea that an SCC should be invariant under the perturbation of extreme information.
In statistical contexts, outliers are typically identified through criteria derived endogenously from the data.
In our setting, however, the set of alternatives is not an interval of real numbers, and the ``distances'' between alternatives are not necessarily uniform.
Consequently, applying standard statistical criteria is inappropriate.
We therefore simplify the treatment by disregarding the observations at both ends as outliers.
To avoid arbitrariness in determining the number of excluded observations, we impose the weakest requirement: we adopt an axiom that removes exactly one observation from each extreme.
This idea is inspired by the concept of the ``trimmed mean,'' which calculates the average after discarding extreme data points from both ends. 
A prominent example of research on trimmed means within mechanism design is \cite{Louis2023trimming}, and the specific case where the single highest and lowest reports are excluded is known as the ``Olympic mean." 
While the primary purpose of the trimmed mean is to eliminate extreme opinions, and although this study does not explicitly discuss the trimmed mean itself, both share the underlying spirit of ensuring that the collective decision does not unduly depend on extreme reports.
By definition, one can easily verify that the following axiom implies \textit{peak-onlyness}, which requires that the choice depend solely on the profile of peak alternatives and can be interpreted as an informational efficiency condition.

\begin{itemize}
\item[]  \textbf{Tail invariance}: For any $P, P' \in \mathcal{SP}^n$, if $\tau(P) \setminus \{\underline \tau(P), \overline \tau(P)\}=\tau(P') \setminus \{\underline \tau(P'), \overline \tau(P')\}$, then $F(P)=F(P')$. 

\item[] \textbf{Peak-onlyness}: For any $P,P^\prime \in \mathcal{SP}^n$ such that $\tau(P_i)=\tau(P_i^\prime)$ for all $i\in N$, $F(P)=F(P^\prime)$.
\end{itemize}

\section{Main results}\label{sec_main}

We first show that \textit{peak-robustness} is quite a strong requirement, and we can eventually obtain the median voter rules as the unique rule that meets the criterion with \textit{efficiency} and \textit{tail invariance}.

\begin{theorem}\label{thm_median}
$F$ satisfies \textit{efficiency}, \textit{peak-robustness}, and \textit{tail invariance} if and only if $F$ is the median voter rule.
\end{theorem}

While the median voter rule is conventionally discussed primarily in terms of its strategy-proofness, the key significance of this result lies in the fact that we do not explicitly require strategy-proofness to characterize it. 
Instead, even from a statistical robustness perspective--independent of strategy-proofness--our findings uniquely distinguish the median voter rule within the vast landscape of social choice correspondences.

Next, we examine the implications of \textit{weak peak-robustness} in conjunction with other desirable properties. 
In contrast to Theorem \ref{thm_median}, which uniquely characterizes the median voter rule, we obtain a broader class of SCCs.
Having relaxed the robustness requirement to \textit{weak peak-robustness}, we consider the class of SCCs satisfying additional requirements for endpoints, \textit{endpoint invariance}, and \textit{peak monotonicity}, each of which axioms is motivated by the role of the endpoints as thresholds determining which opinions are discarded.

The first axiom requires that the thresholds (i.e., $l(P)$ and $u(P)$) be independent of the particular distribution of observations lying outside them. 
Intuitively, this means that discarded observations--those below the lower endpoint or above the upper endpoint--affect the outcome only through their number, not through their individual values. 
Without this requirement, the outcome could vary with the distribution of these tail observations even when their total number remains unchanged.

\begin{itemize}
\item[] \textbf{Endpoint invariance}: For any $P, P^\prime \in \mathcal{SP}^n$ and any $\gamma(\cdot) \in \{l( \cdot),u( \cdot)\}$, if
\[
|\{ i \in N \mid \tau(P_i) \lessgtr  \gamma(P)\}| = |\{ i \in N \mid \tau(P_i^\prime) \lessgtr \gamma(P) \}|,
\]
then $\gamma(P)=\gamma(P^\prime)$.
\end{itemize}

Although \textit{endpoint invariance} is related to \textit{tail invariance}, the two axioms are independent and capture distinct statistical considerations. 
\textit{Tail invariance} is primarily concerned with the influence of outliers, namely observations located at the extreme ends of the data. 
\textit{Endpoint invariance}, by contrast, concerns observations that are excluded by the selected thresholds. 
Such observations need not be extreme outliers; they are simply those that fall outside the chosen interval. 
Since the distance between alternatives is not commonly understood by the social planner and the voters, making the thresholds depend on the detailed distribution of discarded opinions may lead to dissatisfaction among voters. To avoid this issue, we require that the endpoints depend solely on the number of voters on either side of them, thereby ensuring that they are determined only by information commonly observable to all agents.

We also require the endpoints to satisfy a natural monotonicity condition. 
The following axiom requires directional consistency between changes in an individual's preference and the resulting movement of the socially selected endpoints. 
Specifically, if all other agents' preferences remain fixed and a single agent's peak shifts to the right, then both the lower endpoint $l(P)$ and the upper endpoint $u(P)$ should move weakly to the right as well; they should not move in the opposite direction.

Note that \textit{Peak monotonicity} is implied by \textit{optimistic and pessimistic strategy-proofness} together with \textit{peak-onlyness}.\footnote{This implication is well known for social choice functions. See Lemma 1 of \cite{klaus2020strategy} for a formal proof.}

\begin{itemize}
\item[] \textbf{Peak-monotonicity}: For any $i \in N$ and any $P, P^\prime \in \mathcal{SP}^n$ such that $P_{-i} = P^\prime_{-i}$, if $\tau(P_i) \le \tau(P^\prime_i)$, then $l(P) \le l(P^\prime)$ and $u(P) \le u(P^\prime)$.
\end{itemize}

\begin{theorem}\label{thm_quan}
 $F$ satisfies  \textit{anonymity}, \textit{efficiency}, \textit{weak peak-robustness}, \textit{tail invariance}, \textit{endpoint invariance} and \textit{peak-monotonicity} if and only if  $F$ is a quantile rule.    
\end{theorem}

Note that the quantile rule does not necessarily treat alternatives symmetrically, depending on the values of the parameters $Q_L$ and $Q_U$.
However, it should be noted that we do not always need to treat alternatives with strict symmetry.
Consider, for instance, a social choice setting such as determining a tax rate.
In this scenario, it may be natural for a social planner to desire as high a tax rate as possible, even while taking voters' opinions into account. 
In such cases, rather than treating all opinions symmetrically, a voting rule that places greater weight on upper-side opinions for the sake of tax-funded social administration would be well-justified.
In standard political contexts, however, treating leftist and rightist ideologies asymmetrically would be unacceptable.

By imposing additional symmetric property, \textit{neutrality under reversal}, we can obtain symmetric quantile rules as special cases.
The following axiom is based on the one by \cite{kiguchi2025may}.
Let the reflection mapping $\rho: X \to X$ be such that $\rho(a_u) = a_{K-u+1}$ for all $u\in \{1,\ldots,K\}$.
For any $P_i\in \mathcal{SP}$, let $P_i^-$ be defined as 
\[
xP_iy \Leftrightarrow \rho(x)P_i^-\rho(y), \forall x,y \in X.
\]
Let $P^-=(P_1^-,\ldots, P_n^-)$.
For any subset of alternatives $A \in 2^X\setminus\{\emptyset\}$, let $A^- = \{ \rho(a) \mid a \in A \}$.

Neutrality under reversal requires that the chosen outcome change symmetrically when all voters reverse their preferences.

\begin{itemize}
\item[] \textbf{Neutrality under reversal}: For any $P\in \mathcal{SP}^n$,
\[
F(P^-)=F(P)^-.
\]
\end{itemize}

\begin{proposition}\label{prop_symquan}
$F$ satisfies \textit{efficiency}, \textit{anonymity}, \textit{weak peak-robustness}, \textit{tail invariance}, \textit{endpoint invariance}, \textit{peak-monotonicity} and \textit{neutrality under reversal} if and only if  $F$ is a symmetric quantile rule. 
\end{proposition}

\section{Concluding remarks}\label{sec_conclusion}
In this paper, we introduced a set of criteria for social choice correspondences inspired by the principles of robust statistics, with the aim of designing voting rules that are robust to both the detailed specification of individual preferences and the presence of extreme opinions. 
We showed that quantile rules are the unique class of voting rules satisfying these criteria.
Since quantiles--most notably the median--are well known as robust estimators in statistics, our results identify their counterparts in social choice theory and demonstrate that they share analogous notions of robustness. 
In this sense, we uncover a novel connection between social choice theory and robust statistics.
We conclude by discussing several related issues and possible directions for future research.

\subsection{Strategy-proofness and generalized median correspondences}

While generalized median correspondences constitute a broad class of social choice correspondences and are characterized by a set-valued notion of strategy-proofness \citep{Klaus2020target}, the quantile rule is characterized as a particular subclass based on robustness considerations.\footnote{Several variants of strategy-proofness for social choice correspondences have been extensively studied in the literature; see, for example, \citet{barbera2001strategy}, \citet{ching2002multi}, \citet{ozyurt2008strategy}, \citet{nehring2000monotonicity}, \citet{yi2007monotonicity}, and \citet{brandt2015set}, and the references therein.}

One might consider that strategy-proofness and robustness are closely related properties. 
However, we show that our robustness concepts and strategy-proofness are logically independent.\footnote{Here we focus on the fact that strategy-proofness does not imply robustness. For the converse direction, see the example establishing the logical independence of \textit{peak monotonicity} in Theorem \ref{thm_quan}. Since \textit{peak monotonicity} is implied by strategy-proofness, that example also demonstrates that robustness does not imply strategy-proofness.}
To illustrate this point, consider another well-known subclass of generalized median correspondences, namely the \textbf{target set correspondence} introduced by \citet{Klaus2020target}. 
This correspondence extends the target rule of \citet{Thomson1993target} to the setting of social choice correspondences and is defined as follows.

\begin{itemize}
    \item[] \textbf{Target set correspondence}: There exist $a ,b \in X (a \le b)$, for any $P\in \mathcal{SP}^n$, \\
    $F^{a,b}(P) = \begin{cases}
\{\underline{\tau}(P)\} & \text{if } b < \underline{\tau}(P), \\
\{\overline{\tau}(P)\} & \text{if } a > \overline{\tau}(P), \\
[a, b] \cap [\underline{\tau}(P),\overline{\tau}(P)] & \text{otherwise.}
\end{cases}$
\end{itemize}

We can see that if $|\{j \in N \mid \alpha_j=a\}|=n-1$ and $|\{j \in N \mid \beta_j=b\}|=n-1$, then an (anonymous and efficient) generalized median correspondence $F^{\alpha,\beta}$ is a target set correspondence.
Among the axioms characterizing the quantile rule, target set correspondences satisfy \textit{anonymity}, \textit{efficiency}, \textit{endpoint invariance}, and \textit{peak monotonicity}. 
However, they generally fail to satisfy \textit{weak peak-robustness} and \textit{tail invariance}.
To see this, note first that a target set correspondence does not necessarily select the median voter's peak. 
Hence, by Lemma \ref{lem_condorcet_convex}, it fails to satisfy \textit{weak peak-robustness}. 
Moreover, \textit{tail invariance} is generally violated as well, since the outcome of a target set correspondence may depend on the precise locations of observations in the tails rather than solely on the interval determined by the remaining observations.

\subsection{Even number of voters}
Throughout the paper, we assume that the number of voters is odd so that the median peak is uniquely defined. When $n=2m$ with $m\ge 1$, however, the median need not be unique. 
Let $\tau_{(m)}(P)$ and $\tau_{(m+1)}(P)$ denote the $m$-th and $(m+1)$-st order statistics of the peak profile, which we call the \emph{left median} and \emph{right median}, respectively.

In this case, any social choice correspondence satisfying
\[
F(P)\subseteq \{\tau_{(m)}(P),\tau_{(m+1)}(P)\}
\]
for every $P\in\mathcal{SP}^n$ satisfies efficiency, peak-robustness, and tail invariance. 
Indeed, neither median can be robustly dominated by another alternative, while any majority comparison between the left and right medians can result at most in a tie. 
Similarly, quantile rules with $Q_L+Q_U\le 2m-1$ satisfy all the axioms in Theorem \ref{thm_quan}.

This observation contrasts with the classical strategy-proofness approach. 
When the number of voters is even, the median voter rule is no longer uniquely defined, and a rule that selects either the left or the right median depending on the preference profile need not be strategy-proof. 
By contrast, robustness considerations continue to identify the set of median alternatives as the relevant class of outcomes. 
Thus, robustness remains useful for characterizing desirable voting rules even when a unique median alternative does not exist.

\subsection{A stronger version of peak-robustness}
Peak-robustness requires that every alternative selected by an SCC be robustly undominated by any unselected alternative.
Strengthening this requirement so that every selected alternative is robustly undominated by \textit{all other} alternatives yields the following stronger notion.

\begin{itemize}
\item[] \textbf{Strong peak-robustness}
For any $P \in \mathcal{SP}^n$ and $x \in F(P)$, there is no $P' \in \mathcal{SP}^n$ with $\tau(P')=\tau(P)$ and $a \neq x$ such that $|\{i \in N \mid a P'_i x \}|>|\{i \in N \mid x P'_i a \}|$.
\end{itemize}

It is easy to see that the median rule satisfies \textit{strong peak-robustness}, since the median coincides with the Condorcet winner on the single-peaked domain.
The converse is also true. 
To see this, suppose that an SCC selects an alternative other than the median. 
Since the median is the Condorcet winner, any such alternative is robustly dominated by the median. 
Hence, no SCC satisfying \textit{strong peak-robustness} can select a non-median alternative. 
It follows that any SCC satisfying \textit{strong peak-robustness} must be single-valued and select the median alternative for every preference profile.\footnote{We thank Yuki Tamura for suggesting a clarification of the implications of this stronger notion of peak-robustness.}

\subsection{Preference domain}
Throughout our analysis, we have focused on the single-peaked domain. 
A natural direction for future research is to identify broader preference domains on which similar results continue to hold.

\citet{barbera2011top} introduced \textit{top monotonicity}, a domain restriction weaker than several classical conditions, including single-peakedness, single-plateauedness, and single-crossingness. 
They show that top monotonicity is sufficient for the existence of a Condorcet winner and, moreover, that the Condorcet winner coincides with a median peak. 
Consequently, on this domain, peak-robustness implies that an SCF must include the median peak. 
This observation suggests that it may be possible to characterize the median voter rule on the top-monotonicity domain using robustness-based axioms analogous to those employed in Theorem \ref{thm_median}.

This possibility stands in sharp contrast to the strategy-proof foundations of median voter rules. \citet{chatterji2023taxonomy} identify the preference domains on which strategy-proof social choice functions satisfy peak-onlyness, a necessary property of generalized median voter schemes, and show that these domains are essentially single-peaked.\footnote{See also \citet{chatterji2011tops} and \citet{weymark2008strategy,weymark2011unified} for comprehensive discussions of the relationship between strategy-proofness and peak-onlyness.} 
Thus, while the strategy-proof foundation of median voter rules appears to be closely tied to single-peakedness, robustness-based foundations may extend to broader domains.
Determining whether robustness characterizations of median voter rules or generalized median correspondences can indeed be established beyond the single-peaked domain remains an interesting direction for future research.

\subsection{Tail invariance and responsiveness}
Our characterization of quantile rules may be viewed as a refinement of the characterization of generalized median correspondences $F^{\alpha, \beta}$: once the latter class has been identified, our axioms uniquely determine the admissible values of the parameters $\alpha$ and $\beta$.
In this sense, our approach is closely related to that of \citet{kiguchi2025may}, who identify conditions that uniquely determine the median voter rule within the class of generalized median voter rules.

Specifically, \citet{kiguchi2025may} introduce the following notion of \textit{weak positive responsiveness}, which requires a voting rule to respond monotonically to voters’ preference shifts in a common direction.
\begin{itemize}
\item[] \textbf{Weak positive responsiveness}:
\begin{enumerate}
  \item For any $P\in \mathcal{SP}^n$ satisfying $\tau(P_i) \le a_{K-1}$ for all $i\in N$, if $l(P), u(P)\le K-1$,
\[
l(P)<l(P^{+1}), u(P)<u(P^{+1}).
\]
  \item For any $P\in \mathcal{SP}^n$ satisfying $\tau(P_i) \ge a_{2}$ for all $i\in N$, if $l(P), u(P)\ge 2$,
\[
l(P)>l(P^{-1}), u(P)>u(P^{-1}).
\]
\end{enumerate}
where, for any $P_i\in \mathcal{SP}$ satisfying $\tau(P_i)\le a_{K-1}$, let $P_i^{+1}$ be defined as
\[
a_u P_i a_v \Leftrightarrow a_{u+1} P^{+1}_i a_{v+1}, \forall u,v\le K-1,~\text{and}~ a_u P_i^{+1} a_1, \forall u\ge 2, 
\]
and, for any $P_i\in \mathcal{SP}$ satisfying $\tau(P_i)\ge a_2$, let $P_i^{-1}$ be defined as
\[
a_u P_i a_v \Leftrightarrow a_{u-1} P^{+1}_i a_{v-1}, \forall u,v\ge 2,~\text{and}~a_u P_i^{+1} a_K, \forall u \le K-1.
\]
\end{itemize}

While a generalized median correspondence does not necessarily satisfy \textit{weak positive responsiveness}, every quantile rule does. 
Moreover, a generalized median correspondence satisfies \textit{weak positive responsiveness} only if $\alpha_j,\beta_j\in\{a_1,a_K\}$ for every $j\in\{1,\ldots,n-1\}$.

To see this, suppose that there exists some $j\in\{1,\ldots,n-1\}$ such that $\alpha_j\in X\setminus\{a_1,a_{K-1},a_K\}$. Let $\alpha_j=a_t$ and consider a profile $P\in\mathcal{SP}^n$ satisfying
\[
|\{i \in N\mid \tau(P_i)<a_t\}|+|\{j \in N\mid \alpha_j\le a_t\}|=m+1.
\]
Then,
\[
l(P)=\operatorname{med}\{\tau(P_1),\ldots,\tau(P_n),\alpha_1,\ldots,\alpha_{n-1}\}=a_t.
\]
Now consider the profile $P^{+1}$. 
By construction, we still have $l(P^{+1})=a_t$, which contradicts \textit{weak positive responsiveness}.
Next, suppose that there exists some $j\in\{1,\ldots,n-1\}$ such that $\alpha_j=a_{K-1}$. Consider a profile $P\in\mathcal{SP}^n$ satisfying
\[
|\{i \in N\mid \tau(P_i)=a_K\}|+|\{j \in N\mid \alpha_j\ge a_K\}|=m+1.
\]
Then
\[
l(P)=\operatorname{med}\{\tau(P_1),\ldots,\tau(P_n),\alpha_1,\ldots,\alpha_{n-1}\}=a_{K-1}.
\]
Now consider the profile $P^{-1}$. 
By construction, we again obtain $l(P^{-1})=a_{K-1}$, contradicting \textit{weak positive responsiveness}. Therefore, $\alpha_j\in\{a_1,a_K\}$ for every $j\in\{1,\ldots,n-1\}$.
 By a symmetric argument applied to $u(P)$, it follows that $\beta_j\in\{a_1,a_K\}$ for every $j\in\{1,\ldots,n-1\}$.

The above argument illustrates that generalized median correspondences that are insensitive to uniform shifts in voters' preferences necessarily rely on fixed boundary parameters and are therefore vulnerable to the influence of extreme opinions. 
While \citet{kiguchi2025may} employ weak positive responsiveness (in its single-valued form) to characterize the median voter rule within the class of generalized median voter rules, our characterization uses \textit{tail invariance} to play a similar role. 
Both axioms serve to eliminate generalized median rules whose outcomes are excessively influenced by boundary parameters rather than by the distribution of voters' peaks.

\begin{center}
\Large{{\bf Appendix}}
\end{center}

\numberwithin{definition}{section}
\numberwithin{theorem}{section}
\numberwithin{lemma}{section}
\numberwithin{proposition}{section}
\numberwithin{corollary}{section}
\numberwithin{example}{section}
\renewcommand{\theequation}{\thesection.\arabic{equation}}
\appendix

\section{Omitted Proofs}
Before presenting the proofs of our main results, we make the following observation regarding \textit{efficiency}.

\begin{lemma}\label{lem_efficient}
$F$ satisfies \textit{efficiency} if and only if $F(P) \subseteq  [\underline \tau(P), \overline \tau(P)]$ for any $P \in \mathcal{SP}^n$.
\end{lemma}

\begin{proof}
If part is obvious, so it suffices to show the only if part.
Suppose that $a \in F(P)$ such that $a \notin [\underline \tau(P), \overline \tau(P)]$ for some $P \in \mathcal{SP}^n$ and $a \in X$.
Without loss of generality, we assume that $a<\underline\tau(P)$.
Then, since each $P_i$ is single-peaked, we have $\underline\tau(P) P_i a$ for any $i \in N$, which contradicts \textit{efficiency}.
\end{proof}

\subsection{Proof of Theorem \ref{thm_median}}
\begin{proof}[Proof of Theorem \ref{thm_median}]
Since if part is obvious, we show the only if part.
Let $F$ be a social choice correspondence that satisfies  \textit{tail invariance}, \textit{peak-robustness}, and \textit{efficiency}.
By Lemma \ref{lem_condorcet_convex}, $F$ must be an interval that contains the Condorcet winner, that is, for any $P \in \mathcal{SP}^n$, there exists  $a(P) \le C(P) \le b(P) \in X$ such that $F(P)=[a(P), b(P)]$.
We show that, if $a(P)<b(P)$, we have $a(P)=a_1$ and $b(P)=a_K$.
Indeed, if $a(P) \neq a_1$, for any $a \in [a_1, a(P)]$, there exists $P'$ with $\tau(P')=\tau(P)$ such that $a P'_i b(P)$ for any $i \in N$ with $\tau_i(P') \le \tau_i(P) \le C(P)$.
Therefore, by \textit{peak-robustness}, we have $a \in F(P)$ for any  $[a_1, a(P)]$.
Symmetrically, we can conclude that $b \in F(P)$ for any $b \in [b(P), a_K]$.

By the above arguments, we have shown that $F(P)=\text{med}\{\tau(P)\}$ or $X$ for any $P \in \mathcal{SP}^n$.
By Lemma \ref{lem_efficient},  if $F(P)=X$ for some $P \in \mathcal{SP}^n$, it must be that $\underline \tau(P)=a_1$ and $\overline \tau(P)=a_K$.
Hence, for $P' \in \mathcal{SP}^n$ with $\tau(P) \setminus \{\underline \tau(P), \overline \tau(P)\}=\tau(P') \setminus \{\underline \tau(P'), \overline \tau(P')\}$ and $\underline \tau(P') \neq a_1$ or $\overline \tau(P') \neq a_K$, we must have
$F(P)=X \neq \text{med}\{\tau(P')\}=F(P')$, which contradicts \textit{tail invariance}.
Therefore, $F(P)=\text{med}\{\tau(P)\}$ for any $P \in \mathcal{SP}^n$, which completes the proof.
\end{proof}

\subsection{Proof of Theorem \ref{thm_quan}}
The proof of Theorem \ref{thm_median} relies on the following lemma, which characterizes the special case of generalized median correspondences.

\begin{lemma}\label{lem_gmc}
$F$ satisfies \textit{anonymity}, \textit{efficiency}, \textit{weak peak-robustness}, \textit{peak-onlyness}, \textit{endpoint invariance} and \textit{peak-monotonicity} if and only if  $F$ is a generalized median correspondence $F^{\alpha, \beta}$ where $m \le |\{j \in N\mid \alpha_j=a_1\}| $ and $m \le |\{j \in N\mid \beta_j=a_K\}|$.
\end{lemma}

\begin{proof}
We first show the necessity of the axioms.
\textit{Anonymity} and \textit{peak-onlyness} follow from the definition of generalized median correspondences. 
Hence, it suffices to show that any rule in this class satisfies \textit{efficency}, \textit{weak peak-robustness}, \textit{endpoint invariance} and \textit{peak-monotonicity}.
\\

\underline{\textit{Efficiency}}: 
By Lemma \ref{lem_efficient}, it suffices to show $F(P) \subseteq  [\underline \tau(P), \overline \tau(P)]$ for any $P \in \mathcal{SP}^n$.
Suppose that $F(P) \nsubseteq  [\underline \tau(P), \overline \tau(P)]$.
This implies $l(P) < \underline \tau(P)$ or $u(P) > \overline \tau(P)$.
Assume that $l(P) < \underline \tau(P)$.
By the definition of $F$, $l(P) = \text{med}\{\tau(P), \alpha_1,\ldots,\alpha_{n-1}\}$.
Therefore, $l(P) \in \{\alpha_1,\ldots,\alpha_{n-1}\}$.
Suppose that $l(P) = \alpha_t < \underline \tau(P)$, $t \in \{1,\ldots,n-1\}$.
Then, there exists $i \in N$ such that $\tau(P_i) < \alpha_t$ because $l(P)$ is the $n$-th alternative of $\{\tau(P), \alpha_1,\ldots,\alpha_{n-1}\}$ by definition.
However, this contradicts $\alpha_t < \underline \tau(P)$.
Similarly, we obtain $u(P) \le \overline \tau(P)$
Therefore, $F(P) \subseteq  [\underline \tau(P), \overline \tau(P)]$, which means $F$ is \textit{efficiency}.
\\

\underline{ \textit{Weak peak-robustness}}:
By Lemma \ref{lem_condorcet_convex}, it suffices to show that $F$ satisfies \textit{set-valued Condorcet consistency} and \textit{convexity} to prove $F$ satisfies \textit{weak peak-robustness}.
Convexity follows from the definition of generalized quantile rules.
We now show that $F$ satisfies \textit{set-valued Condorcet consistency}.
It can be verified that $|\{i \in N \mid \tau(P_i) \le C(P)\}| \ge m+1$ and by the definition of $F$, $|\{j \in N  \mid \alpha_j = a_1\}| \ge m$, so $|\{i \in N \mid \tau(P_i) \le C(P)\}| + |\{j \in N  \mid \alpha_j = a_1\}| \ge 2m+1 = \text{med}\{1,2,\ldots,2n-1\} $.
This leads to $l(P) = \text{med}\{\tau(P), \alpha_1,\ldots,\alpha_{n-1}\} \le \text{med}\{\tau(P)\}$.
Similarly, we obtain $u(P) = \text{med}\{\tau(P), \beta_1,\ldots,\beta_{n-1}\} \ge \text{med}\{\tau(P)\}$.
By \textit{convexity} of $F$, we have $\text{med}\{\tau(P)\} = C(P) \in F(P)$.
Therefore, $F$ satisfies \textit{set-valued Condorcet consistency}.
\\

\underline{\textit{Endpoint invariance}}: 
Take any $P \in \mathcal{SP}^n$.
Let $|\{x \in  \{\tau(P), \alpha_1,\ldots,\alpha_{n-1}\} \mid  x < l(P)\}| = n_L$, $|\{x\in \{\tau(P), \alpha_1,\ldots, \alpha_{n-1}\} \mid x = l(P)\}| = n_M$ and $|\{x \in  \{\tau(P), \alpha_1,\ldots,\alpha_{n-1}\} \mid  x > l(P)\}| = n_R$.
Since $l(P)$ is the $n$-th alternative of $\{\tau(P), \alpha_1,\ldots,\alpha_{n-1}\}$ by definition, we obtain $n_L \le n-1$, $n_R \le n-1$, $n_M + n_R \ge n$.
Obviously, $n_M >0$.
Now, we consider $P^\prime \in \mathcal{SP}^n$ such that $|\{i \in N \mid \tau(P_i^\prime) < l(P)\}| = |\{i \in N \mid \tau(P_i) < l(P)\}|$ and $|\{i \in N \mid \tau(P_i^\prime) > l(P\}| = |\{i \in N \mid \tau(P_i) > l(P)\}|$.
Since $\alpha_1,\ldots,\alpha_{n-1}$ are fixed alternatives, we get $|\{x \in  \{\tau(P^\prime), \alpha_1,\ldots,\alpha_{n-1}\} \mid  x < l(P)\}| = n_L$ and $|\{x \in  \{\tau(P^\prime), \alpha_1,\ldots,\alpha_{n-1}\} \mid  x > l(P)\}| = n_R$.
Hence, it follows that $|\{x \in  \{\tau(P^\prime), \alpha_1,\ldots,\alpha_{n-1}\} \mid  x = l(P)\}| = n_M >0$. 
Given that $l(P^\prime)$ is the $n$-th alternative of $\{\tau(P^\prime), \alpha_1,\ldots,\alpha_{n-1}\}$, and $n_L \le n-1 , n_R \le n-1, n_M + n_R \ge n$, we have $l(P^\prime) = l(P)$.
The argument of $u$ is symmetric.
\\

\underline{\textit{Peak-monotonicity}}: 
Take any $P \in \mathcal{SP}^n$.
We consider $P^\prime \in \mathcal{SP}^n$ such that $P_{-i} = P^\prime_{-i}$ and $\tau(P_i) \le \tau(P_i^\prime)$.
Obviously, we obtain $l(P) = \text{med}\{\tau(P), \alpha_1,\ldots,\alpha_{n-1}\} \le \text{med}\{\tau(P^\prime),\alpha_1,\ldots,\alpha_{n-1}\} = l(P^\prime)$.
This implies \textit{peak-monotonicity} of $l$.
The argument of $u$ is symmetric.
\\

We now prove sufficiency. 
Let $F$ be an SCC that satisfies all the axioms.
By Lemmas \ref{lem_condorcet_convex} and \ref{lem_efficient}, we can write $F(P)=[l(P),u(P)] \subset[\underline \tau(P),\overline \tau(P)]$ for any $P \in \mathcal{SP}^n$.
Note that $C(P)=\text{med}\{\tau(P_1), \ldots, ,\tau(P_n)\} \in F(P)$.
Therefore, we can also see that 
\[
\underline \tau(P) \le l(P)\le \text{med}\{\tau(P_1), \ldots, ,\tau(P_n)\} \le u(P) \le \overline \tau(P)
\]
for any $P \in \mathcal{SP}^n$.
By \textit{peak-onlyness} of $F$, both $l$ and $u$ can be considered as single-valued SCCs $l,u: \mathcal{SP}^n\rightarrow X$ that satisfy \textit{peak-onlyness}. 

In what follows, we show that both $l(P)$ and $u(P)$ can be represented as generalized median median voting rules.
By Proposition \ref{prop_Moulin}, it suffices to show that a single-valued SCC $l,u$ satisfies \textit{anonymity}, \textit{efficiency}, and \textit{strategy-proofness}.
\\

\underline{\textit{Anonymity} of $l$ and $u$}:
By \textit{anonymity} of $F$, for $P \in \mathcal{SP}^n$ and any pemutation $\sigma$ on $N$, we have $[l(P),u(P)]=F(P)=F(P^{\sigma})=[l(P^{\sigma}),u(P^{\sigma})]$, which implies that $l(P)=l(P^{\sigma})$ and $u(P)=u(P^{\sigma})$.
\\

\underline{\textit{Efficiency} of $l$ and $u$}:
Since $l(P) \in [\underline \tau(P), \text{med}\{\tau(P_1), \ldots, ,\tau(P_n)\}]$ for any $P \in \mathcal{SP}^n$, the voter $i$ with $\tau(P_i)=\underline \tau(P)$ satisfies $l(P)P_i x$ for any $x>l(P)$.
Similarly,  for any $x<l(P)$, there exists $i \in N$ such that $\text{med}\{\tau(P_1), \ldots, ,\tau(P_n)\} \le \tau(P_i)$, which implies $l(P) P_i x$.
Therefore, $l$ satisfies \textit{efficiency}.
Similar reasoning also shows that $u$ satisfies \textit{efficiency} as well.
\\

\underline{\textit{Strategy-proofness of $l$ and $u$}}:
Take any $P \in \mathcal{SP}^n$.
First, let us consider a voter $i \in N$ with $\tau(P_i) > l(P)$.
By \textit{endpoint invariance}, we have $l(P) = l(P_i^\prime,P_{-i})$ for any $P^\prime \in \mathcal{SP}^n$ where $l(P) \le \tau(P^\prime _i)$.
Next, by \textit{peak-monotonicity}, we have $l(P_i^\prime,P_{-i})\le l(P)$ for any $P^\prime \in \mathcal{SP}^n$ where $\tau(P^\prime _i) < l(P)$.
By single-peakedness, for any $x<l(P)$, we must have $l(P) P_i x$.
Hence, in this case, $l(P) P_i l(P^\prime_i,P_{-i})$.
Consequently, any voter $i \in N$ whose peak satisfies $\tau(P_i) > l(P)$ has no incentive to misreport their preference.

By a symmetric argument, we obtain that a voter $i \in N$ with $\tau(P_i) < l(P)$ has no incentive to misreport their preference.
Obviously, a voter $i \in N$ with $\tau(P_i) = l(P)$ has no incentive to manipulate.
Summarizing the above arguments, we obtain that $l$ satisfies \textit{strategy-proofness}.
The argument for $u$ is symmetric.
\\

By Proposition \ref{prop_Moulin}, we have shown that, for any $P \in \mathcal{SP}^n$,
\[
l(P)=\text{med}\{\tau(P_1), \ldots, \tau(P_n),\alpha_1,\ldots, \alpha_{n-1}\},~~u(P)=\text{med}\{\tau(P_1), \ldots, \tau(P_n),\beta_1,\ldots, \beta_{n-1}\}
\] 
for some  $(\alpha,\beta)=(\alpha_1,\ldots,\alpha_{n-1},\beta_1,\ldots,\beta_{n-1}) \in X^{2(n-1)}$.
As we have seen that $\underline \tau(P) \le l(P)\le \text{med}\{\tau(P_1), \ldots, ,\tau(P_n)\} \le u(P) \le \overline \tau(P)$, we must have 
$|\{j\mid \alpha_j=a_1\}| \ge m$ and $|\{j\mid \beta_j=a_K\} \ge m|$, which completes the proof.
\end{proof}

\begin{proof}[Proof of Theorem \ref{thm_quan}]
We first show if part.
By Lemma \ref{lem_gmc}, it suffices to show that  \textit{tail invariance} are satisfied.
To show $F$ satisfies \textit{tail invariance}, we consider any voter $i \in N$. 
For any $P \in \mathcal{SP}^n$ , consider any $P^\prime \in \mathcal{SP}^n$ where  $\tau(P) \setminus \{\underline \tau(P), \overline \tau(P)\}=\tau(P') \setminus \{\underline \tau(P'), \overline \tau(P')\}$.
For any $Q \in \{1,\ldots,m\}$, $\tau_{(Q+1)} \neq \tau_{(1)} \text{ or }\tau_{(n)}$  and $\tau_{(n-Q)} \neq \tau_{(n)} \text{ or } \tau_{(n)}$.
Therefore, by the definition of $F$, $F(P) = F(P^\prime)$.

We now prove only if part.
By Lemma \ref{lem_gmc}, since \textit{tail-invariance} implies \textit{peak-onlyness}, according to Lemma \ref{lem_gmc}, it is necessary that $F$ takes the form of a  generalized median correspondence, $F^{\alpha,\beta}$ where $m \le |\{j \in N\mid \alpha_j=a_1\}| $ and $m \le |\{j \in N\mid \beta_j=a_K\}|$.
It suffices to show that $(\alpha,\beta) \in X^{2(n-1)}$ satisfies $\alpha_j, \beta_j \in \{a_1,a_K\} $ for any $j\in\{1, \ldots ,n-1\}$, $m \le |\{j \in N\mid \alpha_j=a_1\}| < 2m$ and $m \le |\{j \in N\mid \beta_j=a_K\}|< 2m$.

First, we show that $(\alpha,\beta) \in X^{2(n-1)}$ satisfies $\alpha_j, \beta_j \in \{a_1,a_K\} $ for any $j\in\{1, \ldots ,n-1\}$.
Suppose that there exists $j \in \{1,\ldots,n-1\}$ such that $\alpha_j \in X \setminus \{a_1,a_K\}$. 
Let such $\alpha_j$ be $a_t$.
Consider the profile $P \in \mathcal{SP}^n$ such that  $|\{j \in N \mid \alpha_j < a_t\}| = 2m$ and $\underline{\tau}(P)=a_t$.
In this case, we obtain $l(P) = a_t$.
For such $P$, we consider $P^\prime \in \mathcal{SP}^n$ such that $\underline{\tau}(P^\prime)=a_{t-1}$ and $\tau(P) \setminus \{\underline \tau(P), \overline \tau(P)\}=\tau(P') \setminus \{\underline \tau(P'), \overline \tau(P')\}$.
By the definition of $F^{\alpha,\beta}$, we obtain $l(P^\prime) = a_{t-1} \neq a_t =l(P)$, which contradicts \textit{tail invariance}. 
By similar logic, $\alpha_j \in  \{a_1,a_K\} \ \forall j\in \{1,\ldots,n-1\}$ can also be shown.
Hence, we proved $(\alpha,\beta) \in X^{2(n-1)}$ satisfies $\alpha_j, \beta_j \in \{a_1,a_K\} $ for any $j\in\{1, \ldots ,n-1\}$.

 Now we show that $(\alpha,\beta) \in X^{2(n-1)}$ satisfies $|\{j \in N\mid \alpha_j=a_1\}| \neq 2m$ and $|\{j \in N\mid \beta_j=a_K\}| \neq 2m$.
Since $F$ satisfies \textit{tail invariance}, $l(P) \neq \tau_{(1)}$ and $u(P) \neq \tau_{(n)}$ for any $P \in \mathcal{SP}^n$.
Based on the discussion above, note that the fixed alternatives $(\alpha, \beta)$ must be positioned as $\alpha_j, \beta_j \in \{a_1,a_K\} $ for any $j\in\{1, \ldots ,n-1\}$, $m \le |\{j \in N\mid \alpha_j=a_1\}| $ and $m \le |\{j \in N\mid \beta_j=a_K\}|$.
In this case, by the definitions of $l(P)$ and $u(P)$, we have $l(P) = \tau_{(|\{j \in N \mid \alpha_j= a_K\}| +1)}$,  and $u(P)=\tau_{(n-|\{j \in N \mid \beta_j= a_1\}|)}$.
Hence, for all $P \in \mathcal{SP}^n, \ l(P) = \tau_{(1)}$ if and only if   $|\{j \in N \mid \alpha_j = a_1\}| = 2m$.
Similarly, $u(P) = \tau_{(n)}$ if and only if $|\{j \in N \mid \beta_j = a_K\}| = 2m$.
Therefore, we obtain $|\{j \in N \mid \alpha_j = a_1\}| , |\{j \in N \mid \beta_j = a_K\}| < 2m$, which completes the proof.
\end{proof}

\subsection{Proof of Proposition \ref{prop_symquan}}

\begin{proof}[Proof of Proposition \ref{prop_symquan}]
We prove necessity first.
By Lemma \ref{lem_gmc}, it suffices to show that \textit{neutrality under reversal} is satisfied.

To show $F$ satisfies \textit{neutrality under reversal}, we consider any voter $i \in N$ and a given profile $P \in \mathcal{SP}^n$.
It suffices to show that $l(P) = \rho (u(P^-))$ and $u(P) = \rho (l(P^-))$ hold for $P^-$ corresponding to this $P$.
Without loss of generality, we assume that $|\{j \in N\mid \alpha_j=a_1\}|=|\{j \in N \mid \beta_j=a_K\}|= Q \ (m \le Q \le 2m)$.
In this case, by the definition of $F$, $l(P)=\tau_{(2m+1-Q)}(P) $ and $u(P) = \tau_{(Q+1)}(P)$.
Next, for this $P$, consider $P^-$.
Similarly, we obtain $l(P^-) = \tau_{(2m+1-Q)}(P^-) $ and $u(P^-) = \tau_{(Q+1)}(P^-)$.
By the definition of $\rho$, it is easy to see $\tau_{(2m+1-Q)}(P) = \rho(\tau_{(Q+1)}(P^-)) $ and $\tau_{(Q+1)}(P) = \rho(\tau_{(2m+1-Q)}(P^-)) $.
Hence, we obtain $l(P) = \rho (u(P^-))$ and $u(P) = \rho (l(P^-))$.

We now prove sufficiency. 
Let $F$ be an SCC that satisfies all the axioms.
By Theorem \ref{thm_quan}, it is necessary that $F$ is a quantile rule, $F^{Q_L,Q_U}$, where $F^{Q_L,Q_U}$ is a generalized median correspondence $F^{\alpha,\beta}$ whose fixed alternatives $(\alpha,\beta) \in X^{2(n-1)}$ satisfy $\alpha_j, \beta_j \in \{a_1,a_K\} $ for any $j\in\{1, \ldots ,n-1\}$, $m \le |\{j \in N\mid \alpha_j=a_1\}| < 2m$ and $m \le |\{j \in N\mid \beta_j=a_K\}|< 2m$.
It is remain to show that $(\alpha,\beta) \in X^{2(n-1)}$ satisfies $\alpha_j, \beta_j \in \{a_1,a_K\} $ for any $j\in\{1, \ldots ,n-1\}$, $m \le |\{j \in N\mid \alpha_j=a_1\}| < 2m$ and $m \le |\{j \in N\mid \beta_j=a_K\}|< 2m$, and $|\{j \in N\mid \alpha_j=a_1\}|=|\{j \in N\mid \beta_j=a_K\}|$.
By this parametric specification, we can see that $F^{\alpha}(P)=\tau_{(|\{j \in N \mid \alpha_j=a_K\}|+1)}(P)$ and $F^{\beta}(P)=\tau_{(n-|\{j \in N\mid \beta_j=a_1\}|)}(P)$.
By setting $Q=|\{j\mid \alpha_j=a_K\}|=|\{j\mid \beta_j=a_1\}|$, we obtain a symmetric quantile rule.

Seeking a contradiction, suppose that $|\{j \in N \mid\alpha_j = a_1\}| \neq |\{j \in N \mid\beta_j = a_K\}|$. 
Without loss of generality, we assume $|\{j \in N \mid\alpha_j = a_1\}| > |\{j \in N \mid\beta_j = a_K\}|$.
Let $P \in \mathcal{SP}^n$ be $|\{i \in N \mid \tau(P_i) = a_1\}| + |\{j \in N \mid \alpha_j = a_1\} |= m+1$.
In this case, $l(P) = \text{med}\{\tau(P_1),\ldots,\tau(P_n),\alpha_1,\ldots,\alpha_{n-1}\} = a_1$.
For $P$ as described above, consider $P^-$.
By the construction of $P$, we obtain $u(P)<a_K$.
This contradicts neutrality under reversal.
Therefore, $|\{j \in N \mid \alpha_j = a_1\}| = |\{j \in N \mid \beta_j = a_K\}|$, which completes the proof.
\end{proof}

\section{Independence of axioms}
Non-redundancy of the axioms in Theorems \ref{thm_median} and \ref{thm_quan} are verified as follows.

 \subsection*{Non-redunduncy  of the axioms in Theorem \ref{thm_median}}

\begin{description}
    \item[] \underline{\textit{Efficiency}}: For any $P \in \mathcal{SP}^n$, let
    \[
    F(P) = X.
    \]
    This rule satisfies all the axioms in Theorem \ref{thm_median} except \textit{efficiency}.

    \item[] \underline{\textit{Peak-robustness}}: For any $P \in \mathcal{SP}^n$, let
    \[
    F(P) = \tau(P) \setminus \{\underline \tau(P), \overline \tau(P)\}.
    \]
     This rule satisfies all the axioms in Theorem 
      \ref{thm_median} except \textit{peak-robustness}.

    \item[] \underline{\textit{Tail invariance}}: For any $P \in \mathcal{SP}^n$, let
        \[
    F(P) = 
    \begin{cases} 
    \text{med}\{\tau(P)\} & \text{if } \underline{\tau}(P) =a_1 \text{ and } \overline{\tau}(P)=a_K, \\
    X & \text{otherwise}.
    \end{cases}
    \]
   This rule satisfies all the axioms in Theorem 
      \ref{thm_median} except \textit{tail invariance}.
      
\end{description}

\subsection*{Non-redunduncy  of the axioms in Theorem \ref{thm_quan}}

\begin{description}
    \item[] \underline{\textit{Anonymity}}: Fix a voter $i$ and suppose that $\tau(P_i)=\tau_{(t)}(P)$. For any $P \in \mathcal{SP}^n$,
    \[
    F(P) = 
    \begin{cases} 
    [\tau_{(2)}(P), \text{med}\{\tau(P)\}] & \text{if } t = 1, \\
    [\tau_{(t)}(P),\text{med}\{\tau(P)\}] & \text{if }2 \le t \le m,
    \\
    \text{med}\{\tau(P)\} & \text{if }  t = m+1,\\
    [\text{med}\{\tau(P)\}, \tau_{(t)}(P)] & \text{if } m+2 \le t \le 2m,  \\
    [\text{med}\{\tau(P)\},\tau_{(2m)}(P)] & \text{if } t = 2m+1.
    \end{cases}
    \]
    This rule satisfies all the axioms in Theorem \ref{thm_quan} except \textit{anonymity}.

    \item[] \underline{\textit{Efficiency}}: 
    For any $P \in \mathcal{SP}^n$, let
    \[
    F(P) = X.
    \]
    This rule satisfies all the axioms in Theorem \ref{thm_quan} except \textit{efficiency}.

   \item[] \underline{\textit{Weak peak-robustness}}: For any $P \in \mathcal{SP}^n$, let
    \[
    F(P) = \tau(P).
    \]
    This rule satisfies all the axioms in Theorem \ref{thm_quan} except \textit{weak peak-robustness}.

    \item[] \underline{\textit{Tail invariance}}: For any $P \in \mathcal{SP}^n$, let
    \[
    F(P) = [\tau_{(1)}(P), \tau_{(n)}(P)].
    \]
    This rule satisfies all the axioms in Theorem \ref{thm_quan} except \textit{tail invariance}.

    \item[] \underline{\textit{Endpoint invariance}}: Suppose that $med\{\tau(P)\} = a_s$, $\tau_{(2)}(P) = a_{t_1}$, and $\tau_{(n-1)}(P) = a_{t_2}$. 
    For any $P \in \mathcal{SP}^n$, let 
    \[
    F(P) = \left[ a_{\lfloor \frac{t_1 + s}{2} \rfloor}, a_{\lceil \frac{s + t_2}{2} \rceil} \right].
    \]
    This rule satisfies all the axioms in Theorem \ref{thm_quan} except \textit{endpoint invariance}.
    
    \item[] \underline{\textit{Peak-monotonicity}}: For any $P \in \mathcal{SP}^n$, let $F(P) = [l(P),u(P)]$ such that
    \[
    l(P) = 
    \begin{cases}
        \tau_3(P) & \text{if } |\{i \in N \mid \tau(P_i) < \tau_3(P)\}| =  2,  \\
        \tau_2(P) & \text{otherwise}.
    \end{cases}
    \]
   \[
    u(P) = 
    \begin{cases}
        \tau_{n-2}(P) & \text{if } |\{i \in N \mid \tau(P_i) > \tau_{n-2}(P)\}| =  2,  \\
        \tau_{n-1}(P) & \text{otherwise}.
    \end{cases}
    \]
    
    This rule satisfies all the axioms in Theorem \ref{thm_quan} except \textit{peak-monotonicity}.

\end{description}

\bibliography{references}

\end{document}